\documentclass[letterpaper, 10 pt, conference]{ieeeconf}
\IEEEoverridecommandlockouts
\usepackage{ifpdf}
\usepackage{cite} % Orders citations.
\usepackage{url}
\usepackage{hyperref}
\usepackage{algorithmic}
\usepackage[pdftex]{graphicx}
\graphicspath{{./figures/}}
\usepackage{color}
\usepackage{pgf, tikz, pgfplots}
\usepackage{pgfplots}
\usepgfplotslibrary{fillbetween}
\usepgfplotslibrary{groupplots}

\usetikzlibrary{shapes, arrows, automata}
\usetikzlibrary{calc,hobby,decorations}
\usepackage[cmex10]{amsmath}
\usepackage{amsfonts, amssymb}
\usepackage{theorem} 
\usepackage{algorithm}
\usepackage{enumerate}
\usepackage{multirow}
\usepackage{rotating}
\usepackage{subfigure}
\usepackage{tikz}
\usepackage{pgfplots}
\usepgfplotslibrary{fillbetween}
\usepgfplotslibrary{groupplots}
\usepackage{caption}
\input{mysymbol.sty}

\newtheorem{assumption}{\bf Assumption}

\newtheorem{lemma}{\bf Lemma}
\newtheorem{proposition}{\bf Proposition}

\newtheorem{theorem}{\bf Theorem}

\newtheorem{remark}{\bf Remark}

\title{\LARGE \bf Parameter Critic: a Model Free Variance \\ Reduction Method Through Imperishable Samples}
\author{Juan~Cervi\~no, Harshat~Kumar, and Alejandro~Ribeiro%
\thanks{Electrical and Systems Engineering, University of Pennsylvania. \texttt{\small \{jcervino, harshat, aribeiro\}@seas.upenn.edu}}%
}

\begin{document}
\bstctlcite{IEEEexample:BSTcontrol}

%\renewcommand \red[1] {}
%\renewcommand \blue[1] {}

% Headers:

\maketitle

\begin{abstract}
We consider the problem of finding a policy that maximizes an expected reward throughout the trajectory of an agent that interacts with an unknown environment. Frequently denoted Reinforcement Learning, this framework suffers from the need of large amount of samples in each step of the learning process.  To this end, we introduce parameter critic, a formulation that allows samples to keep their validity even when the parameters of the policy change. In particular, we propose the use of a function approximator to directly learn the relationship between the parameters and the expected cumulative reward. Through convergence analysis, we demonstrate the parameter critic outperforms gradient-free parameter space exploration techniques as it is robust to noise. Empirically, we show that our method solves the cartpole problem which corroborates our claim as the agent can successfully learn an optimal policy while learning the relationship between the parameters and the cumulative reward. 
\end{abstract}

%!TEX root = ACC_2020_parameter_critic.tex

\section{Introduction} \label{sec_intro}
Consider the problem of maximizing the performance of an agent interacting with the environment. Reinforcement Learning (RL) provides a framework in which the agent can learn the best action to take by subsequently receiving rewards from the environment with which it interacts \cite{sutton2018reinforcement}. By knowing its state, the agent chooses an action, reaches a new state, and receives a reward. Specifically, the agent's objective is to maximize the cumulative reward throughout a trajectory of states and actions. Due to its success, RL has become an ubiquitous tool that has shown successful results in autonomous driving \cite{isele2018navigating,shalev2016safe}, robotics \cite{he2017deep}, communications \cite{peters2003reinforcement} and many others.

Despite its success, one drawback of RL remains in the inefficient use of data. In short, most RL algorithms have in their core an iterative method that improves the expected discounted returns of the policy. The update step is done by deciding the direction of improvement based on a gradient taken conditioning an expectation on the value of the given policy\cite{sutton2000policy,silver2014deterministic}. These methods are able to reduce variance (e.g. Montecarlo methods), however all the samples obtained in an iteration become useless on the following one. This is a consequence of conditioning on the value of the parameters, which, once they are updated on the learning step, the samples remain futile thereafter.

%Valued for its simplicity, parameter space exploration has been shown empirically to be as effective in solving reinforcement learning problems \cite{mania2018simple} as state-of-the-art methods \cite{schulman2015trust, silver2014deterministic}. By viewing the cumulative reward obtained in a system rollout as an objective evaluation, these methods capture reinforcement learning as a black-box optimization problem \cite{ruckstiess2010exploring}.  \blue{Evolution Strategies (TO DO), Although Bayesian optimization techniques also learn a prior on the data, they are known to fail in high parameter spaces and are therefore mainly use for hyper-parameter selection. \cite{young2020distributed}.}

%We chose to focus on the class of solutions which use Gaussian smoothing and finite-differences optimization methods \cite{ghadimi2013stochastic, nesterov2017random,kalogerias2019zeroth}. \blue{Convergence results are simple and easy to understand \cite{vemula2019contrasting}.} Despite their empirical success, zeroth-order methods in parameter space remain limited to the classical reinforcement learning problem with a finite horizon. 

In this work we introduce \emph{parameter critic}, a function approximator which learns the relationship between the parameters of the agent and the expected reward or objective function. Gradient free optimization is then used directly on the learned function \cite{ghadimi2013stochastic, vemula2019contrasting}. Older samples therefore do not become obsolete, as they are noisy evaluations of the objective function which are used to fit the parameter critic. These \emph{imperishable} samples contrast the samples used to update the $Q$-function in actor-critic methods which are conditioned on the actor parameter, and therefore lose their merit once the actor parameter changes. 

Apart from introducing the parameter critic, there are two other main contributions from this work. 
% \red{Removed this part because it was repetitive from the previous section} There are major contributions in this work.  First, we propose a novel zeroth-order algorithm which relies on a function approximator that directly learns the relationship between the parameter and the objective function. Similar to actor-critic methods \cite{kumar2019sample}, the parameter critic can be though of as a variance reduction technique; however, unlike actor-critic methods, which require a new function approximator each actor update, our algorithm uses all parameters evaluated during training. The algorithm uses all state-action pairs in contrast to policy gradient, Q-learning, or actor-critic convergence results which have an i.i.d. assumption on samples \cite{zhang2019global,dalal2017finite, kumar2019sample}. 
First, we establish finite sample complexity bounds for the proposed algorithm which encompasses all categories of reinforcement learning problems (i.e. finite / infinite horizon, finite / continuous state-action space, stochastic / deterministic policy parametrization, non-linear parametrizations). This analysis captures the relationship between function approximation error and the zeroth-order perturbation which points to feasibility in practice.  Finally, using the random horizon rollout trick of \cite{paternain2018stochastic} to obtain unbiased estimates of the expected performance function ($J$), we propose a method to obtain unbiased samples of an infinite horizon objective. This connection enables the the extension of similar zeroth-order methods to infinite horizon problems\cite{vemula2019contrasting, mania2018simple}.

Our results are corroborated numerically on two separate problems. First, consider an example in one dimension to both build intuition on the parametric critic framework as well as highlight the superior performance of our method in noisy settings. Second, we benchmark our scheme on the canonical cart-pole problem, which elucidates a major advantage to parameter space learning. Namely, decreasing the network size results in faster convergence. These implications admit several directions for future work, which are detailed in the conclusion. 
%!TEX root = ACC_2020_parameter_critic.tex

\section{Reinforcement Learning as Black Box Optimization} \label{sec:background}
We consider the reinforcement learning problem (RL), where an agent moves through a state space $\mathcal{S}$ and takes actions in some action space $\mathcal{A}$. After taking an action, the agent transitions to a new state according to an unknown probability $P^{a}_{s\to s'}:= p(s'\vert (s,a) \in \mathcal{S} \times \mathcal{A})$ at which point a reward is revealed by the environment according to the reward function $R:\mathcal{S} \times \mathcal{A} \to \mathbb{R}$. The agent's task is to accumulate as much reward as possible for some possibly infinite set horizon length $\mathcal{H}$. Formally this problem can be encapsulated as a Markov decision process (MDP) as a tuple $(\mathcal{S}, \mathcal{A}, P, R, \gamma)$, where the constant $\gamma \in (0,1)$ is a discount factor which determines how much the future rewards matter to the behavior of the agent. 

In particular, we consider the case where the policy, which can be either stochastic $\pi:\mathcal{S}\to P(\mathcal{A})$ or deterministic $\pi:\mathcal{S}\to \mathcal{A}$, is parametrized by some $\theta \in \Theta \subset \mathbb{R}^p$. For consistency in notation, we let both the stochastic and deterministic action be denoted by $a_t \sim \pi_\theta(s_t)$.  We define the \emph{value} function $V:\mathcal{S} \to \mathbb{R}$ as the expected accumulated rewards throughout a trajectory based on the horizon length,
\begin{align}
    V_H(s) &:= \mathbb{E}_\tau \left[\sum_{t = 1}^H R(s_t, a_t \sim \pi_\theta(s_t))\Big|s_0=s \right], \\
    V_\infty(s) &:= \mathbb{E}_\tau \left[\sum_{t = 1}^\infty  \gamma^{t-1} R(s_t, a_t \sim \pi_\theta(s_t))\Big|s_0=s\right], \\ \nonumber
\end{align}
where the expectation is taken with respect to the trajectory $\tau = \left\{s_0, a_0, s_1, a_1, \dots  \right\}$. The problem of interest is to obtain the parameters $\theta$ that maximize the expected value function $V$. To this end, the objective function $J(\theta)$ that can we written as follows:
%\begin{align} \label{equ:main_problem}
%    \max_\theta &~~  J(\theta):=\nonumber\\
%    &\mathbb{E}_{s \sim \rho^0} \left[ \mathbb{E}\left[\sum_{t = 1}^{\mathcal{H / \infty}} \gamma^t R(s_t, a_t \sim / = \pi_\theta(s_t)) \bigg \rvert s_0 = s \right] \right],
%\end{align}
\begin{equation} \label{equ:main_problem}
\max_\theta \left[ J(\theta) := \mathbb{E}_{s\sim \rho^0} \left[ V(s) \right] \right],    
\end{equation}
where $\rho^0$ is the initial state distribution. For the finite horizon case, unbiased samples of the objective function can be obtained by sampling a state from the initial state distribution $\rho^0$ and rolling out the system $H$ steps \cite{vemula2019contrasting}. For the infinite horizon case, unbiased samples can also be obtained in a similar fashion, but they additionally require sampling the rollout length from a geometric distribution \cite{paternain2018stochastic, zhang2019global} (see Section \ref{sec:infinite_horizon}). In either case, by rolling out the system, an unbiased sample of the objective function $J$ is obtainable. 

As such, similar to the literature \cite{ruckstiess2010exploring, stulp2013robot}, we choose to characterize the parametrized RL problem as a black box optimization problem, where we seek to maximize the objective function $J(\theta)$ using the noisy unbiased samples of the objective function only. Before we proceed to characterize our proposed solution, we first state some fundamental regularity assumptions.

\begin{assumption} \label{asm:bounded_reward} \textbf{(Bounded Objective and Gradient)} There are $\mathcal{J},\beta > 0$ such that for all $\theta \in \Theta$, $|J(\theta)| \leq \mathcal{J}$ and $|\nabla_{\theta}J(\theta)| \leq \beta$. 
\end{assumption}
\begin{assumption} \label{asm:smoothness}\textbf{(Smoothness)} The objective function $J(\theta)$ is differentiable with respect to $\theta$ over the entire domain, and it is L-Lipshitz and G-smooth. That is, there is a number $L < \infty$ such that for every $\theta_1, \theta_2\in \Theta$,
$$\|J(\theta_1) - J(\theta_2)\| \leq L\|\theta_1 - \theta_2\|,$$
and there is a number $G<\infty$ such that for every 
$\theta_1, \theta_2\in \Theta$,
$$\|\nabla J(\theta_1) - \nabla J(\theta_2)\| \leq G\|\theta_1 - \theta_2\|.$$
\end{assumption}

Assumption \ref{asm:bounded_reward} can be achieved by requiring the reward to be bounded, which is standard for deriving performance guarantees in policy search literature \cite{bhatnagar2008incremental,zhang2019global,castro2010convergent} and is typically held in practice \cite{lillicrap2015continuous}. Assumption \ref{asm:smoothness} is also standard for convergence rate guarantees presented in Section \ref{sec:function_approx}.

\subsection{Zeroth Order Optimization}
In this section, we will overview the zeroth order optimization approach for solving the black box optimization problem described previously. We begin by defining the smoothing parameter $\mu >0$ and Gaussian perturbation $\mathbf{u}\sim \mathcal{N}(0, I_p)$. Then, we approximate the gradient of the objective function $J(\theta)$ with respect to the parameters by obtaining two noise rollouts of the system in opposite directions  $\hat{J}(\theta+\mu \mathbf{u}),\hat{J}(\theta-\mu \mathbf{u})$. The $\mu$-smooth objective function surrogate can be defined by 
\begin{equation}
    J_\mu(\theta) := \mathbb{E}_\mathbf{u}\left[J(\theta + \mu \mathbf{u}) \right].
\end{equation}
Using a two-point evaluation of the original objective function $J(\theta)$, we are able to sample an unbiased stochastic estimate of $\nabla_\theta J_\mu (\theta)$. Namely, we recall the key property from \cite{nesterov2017random} later extended to a wider class of functions in \cite{kalogerias2019zeroth}.

\begin{lemma}\label{lem:zeroth_order_grad} \cite[Lemma 2]{kalogerias2019zeroth} For every $\mu > 0$, the $\mu$-smoothed objective function surrogate $J_\mu$ is differentiable, and its gradient admits the representations
\begin{align}
\nabla_\theta J_\mu(\theta) \equiv & \mathbb{E}_{\mathbf{u}\sim \mathcal{N}(0, I_p)}\left[ \frac{J(\theta + \mu \mathbf{u}) - J(\theta )}{\mu} \mathbf{u} \right] \label{equ:asymm}\\
\equiv &\mathbb{E}_{\mathbf{u}\sim \mathcal{N}(0, I_p)}\left[ \frac{J(\theta + \mu \mathbf{u}) - J(\theta - \mu \mathbf{u})}{2\mu} \mathbf{u} \right].\label{equ:symm}\\
\notag
\end{align}
\end{lemma}
For the remainder of the paper, we will be using the symmetric form of the gradient estimate, i.e. \eqref{equ:symm}. Using this zeroth order approximation of the gradient, stochastic gradient descent can be applied to solve problem \ref{equ:main_problem}. The parameter update takes the form 
\begin{equation}\label{equ:param_update_nocritic}
    \theta_{t+1} = \theta_t + \eta \frac{\hat J(\theta + \mu \mathbf{u}) - \hat J(\theta - \mu \mathbf{u})}{2\mu} \mathbf{u},
\end{equation}
where $\hat J(\theta + \mu \mathbf{u})$ and $\hat J(\theta - \mu \mathbf{u})$ are unbiased samples of the objective function. The analysis of \cite{vemula2019contrasting} admits $\delta$-convergence to a stationary point of $J(\theta)$ with $\mathcal{O}(p^2/\delta^3)$ samples.

\subsection{The Infinite Horizon Case} \label{sec:infinite_horizon}

In this section, we show how to use the random horizon trick to extend black box methods to infinite horizon problems. By selecting a random horizon length $H$ from a geometric distribution parameterized by $1-\gamma$, where $\gamma \in (0,1)$ is the discount factor, \cite[Proposition 2]{paternain2018stochastic} showed that the reward collected is an unbiased estimate of the value function. This procedure is detailed in Algorithm \ref{alg:estimateJ}. 

\begin{algorithm}[h]
	\caption{Infinite Horizon $J$-function Sampler \cite{paternain2018stochastic}}
	\begin{algorithmic} \label{alg:estimateJ}
		\REQUIRE Initial state distribution $\rho^0$, Parameters $\theta$
		\STATE Sample $H \sim \textrm{Geom}(1-\gamma)$, Initialize $\hat{J} \leftarrow 0$
		\STATE Draw $s_0 \sim \rho_0$ 
		%from geometric distribution with parameter $\gamma$
		\STATE Select initial action $a_0 \sim \pi_\theta(s_0)$
		\FOR{$t = 1, \dots, H - 1$} 
		\STATE Collect Reward $\hat J \leftarrow \hat J + R(s_t,a_t)$ 
		%\COMMENT{Run the system for first $T_Q$ stages}
		\STATE Advance System $s_{t+1} \sim \mathbb{P}(s'\vert s_t,a_t)$
		\STATE Select Action $a_{t+1} \sim  \pi_\theta(s_{t+1})$
		\ENDFOR
		\STATE Collect Reward $\hat{J} \leftarrow \hat J + R(s_{T}, a_{T})$
		%\STATE Scale $\hat Q = (1-\gamma)\hat J$
	\end{algorithmic}
\end{algorithm}

\begin{proposition}
The infinite horizon $J$-function sampler Algorithm \ref{alg:estimateJ} is unbiased; i.e. $\mathbb{E}[\hat{J}(\theta)]=J(\theta)$.
\end{proposition}

Further, it was shown in \cite[Lemma 2]{kumar2020zeroth} that the procedure has finite variance. Using the fact that Algorithm \ref{alg:estimateJ} admits an unbiased sample of the value function with finite variance, we can trivially extend the convergence results of \cite[Theorem 2]{vemula2019contrasting} to the infinite horizon problem. 

Again, we reiterate the main advantage of our approach relies on the fact that the rollout data collection never loses its validity even if the policy changes. In other words, to train our parameter critic, we can use the whole collection of pairs parameters-expected returns, with which we will reduce the total variance of the gradient estimator. This approach is particularly useful when sampling pairs is noisy, costly, or time consuming.
\begin{remark} 
The way we obtain rollouts for the cumulative expected returns in Algorithm \ref{alg:estimateJ}, only depends on the parameters $\theta$. In the literature \cite{paternain2018stochastic, zhang2019global, kumar2020zeroth}, rollouts are generally done as a function of a fixed state-action pair.  \end{remark}
%!TEX root = ACC_2019.tex

\section{Variance Reduction with Function Approximation} \label{sec:function_approx}

Although unbiased samples of the objective function can be obtained, in general they tend to be noisy due to the possible stochasticity of the system, initial starting position, policy, and random horizon length (for infinite horizon problems). Monte-Carlo rollouts are a naive variance reduction technique, however averaging over several rollouts can be very costly in practice. Instead, we draw motivation from actor-critic algorithms, which use function approximation for a \emph{critic} which acts as a variance reduction method. Unlike actor-critic algorithms, which require finding the optimal action-value function (frequently denoted $Q(s,a)$ which represents the expected return starting from $s$ and taking action $a$)  estimate for \emph{each} actor iteration \cite{kumar2019sample}, we propose directly learning the an objective function surrogate $F:\Theta \times \Omega \to \mathbb{R}$ parametrized by $\omega \in \Omega\subset \mathbb{R}^q$ which is learned over \emph{all} the actor iterates. 

While in practice the parametrized surrogate function $F(\theta, \omega)$ can take the form of any function approximator (such as Radial Basis Function (RBF) kernels, neural networks, or linear function approximation, for example) we choose to focus on a simple single-layer neural network for our theoretical guarantees. This network takes the form 
\begin{equation}\label{eqn:neural_network}
    F(\theta,\omega) := \sum_{i = 1}^N c_i\sigma(\langle a_i, \theta \rangle + b_i)+c_0
\end{equation}
where $N$ is the number of neurons in the network and the point wise non-linear function $\sigma : \mathbb{R}\to \mathbb{R}$ is the sigmoid activation function 
\begin{equation}
    \sigma(x) = \frac{1}{1 + e^{-2x}}.
\end{equation}
We let $\omega = (c_0,a_1, b_1, c_1, \dots a_N, b_N, c_N), \omega \in \reals ^{(p+2)N+1}$ denote the collective parameters. The choice of this specific form of neural network is due to well known the global approximation guarantees \cite{barron1993universal,cybenko1989approximation}, though it is important to note that similar results are possible with deep neural networks with fewer neurons\cite{eldan2016power}. Global function approximation theorems state that for any $\varepsilon^* > 0$, with a sufficiently large network size, there exists an $\omega^*$ such that for all $\theta \in \Theta$,
\begin{equation}\label{equ:critic_bound}
    \max_{\theta \in \Theta} |F(\theta,\omega^*)  - J(\theta) | \leq \varepsilon^*.
\end{equation}
Gradient methods may be used to learn the optimal parameter $\omega^*$. In particular, as loss function $\ell(\cdot,\cdot)$, we propose minimizing the mean squared error (MSE) between the function approximator and the stochastic unbiased estimates which would have otherwise been used for the zeroth order parameter update in \eqref{equ:param_update_nocritic}. That is, given a set of parameters $\mathcal{X} = \left\{\theta_1, \theta_2, \dots,\theta_{2T}\right\}$ and stochastic evaluations of the objective function $\mathcal{Y} = \{\hat J(\theta_1), \hat J(\theta_2), \dots \hat J (\theta_{2T})\}$, we solve the empirical risk minimization with square $2$-norm,
\begin{equation} \label{equ:non_convex}
    w_T \in \arg \min_\omega \frac{1}{2T} \sum_{i = 1}^{2T} \ell(F(\theta_i,\omega) , \hat J(\theta_i)).
\end{equation}
Here we sum to $2T$ due to the parameter-objective function evaluation pair collection prescribed by Algorithm \ref{alg:ZOFA}. Although \eqref{equ:non_convex} is NP-hard in the worst case, we downplay its complexity similar to what was done in \cite{zhang2019global}. In particular, \cite{du2018gradient} and \cite{ge2017learning} have shown that networks with the single layer structure we consider satisfy certain geometric properties which enable the solution to be achieved via trust-region methods \cite{sun2015nonconvex}. As such, we let the function approximator update take the form of a gradient update on the MSE loss on the batch $(\mathcal{X}, \mathcal{Y})$ of size $B$
\begin{equation}\label{eqn:gradient_step}
    \omega_{t+1} = \omega_t - \alpha \hat \nabla_\omega \ell (\mathcal{Y}, F(\mathcal{X},\omega_t)),
\end{equation}
where $\alpha>0$ is the function approximator stepsize. At each time $t$, we define the function approximation error to be 
\begin{equation}
\varepsilon_t :=  \sup_\theta|F(\theta, \omega_t)- J(\theta)|.
\end{equation}
Through the use of gradient methods to update the parameter critic, we assume that $\varepsilon_t$ approaches $\varepsilon^*$ as $t$ grows. The specific criterion for convergence of our algorithm depends on the average of the error iterates, which we define by $\bar \varepsilon_T := 1/T \sum_{t =1}^T \varepsilon_t$, as will be shown in the following subsection.

\begin{algorithm}[t]
	\caption{Gradient Free Parameter Critic Learning}
	\begin{algorithmic}	\label{alg:ZOFA}
		\REQUIRE $\theta_0$, $\omega_0$, $  \alpha, \mu, \eta, \mathcal{X} = \left[\right], \mathcal{Y} = \left[\right]$
		\FOR{$t = 0, 1, \dots$}
        \STATE Sample $\mathbf{u}_t \sim \mathcal{N}(0,I_p)$
		\STATE Obtain noisy estimates $\hat J(\theta_t + \mu \mathbf{u}_t)$ and $\hat J(\theta_t - \mu \mathbf{u}_t)$
		\STATE Add $\theta_t + \mu \mathbf{u}_t$ and $\theta_t -\mu \mathbf{u}_t$ to $\mathcal{X}$
		\STATE Add $\hat J(\theta_t + \mu \mathbf{u}_t)$ and $\hat J(\theta_t - \mu \mathbf{u}_t)$ to $\mathcal{Y}$
%		\STATE Collect batch $(\mathcal{X}_b,\mathcal{Y}_b)\subset(\mathcal{X},\mathcal{Y})$
		\STATE Update function approximator: $$\omega_{t+1} = \omega_t -\alpha \hat{\nabla}_{\omega} \ell(\mathcal{Y},F(\mathcal{X},\omega_T))$$
		\vspace{-5mm}
        \STATE Update actor: $$\theta_{t+1} \leftarrow \theta_t + \eta \frac{F(\theta_t + \mu \mathbf{u}_t,\omega_t) - F (\theta - \mu \mathbf{u}_t,\omega_t)}{2\mu}\mathbf{u}_t$$
		\ENDFOR
	\end{algorithmic}
\end{algorithm}

\subsection{Function Critic Zeroth Order Gradient}

Returning to the RL framework, we seek to obtain the parameters $\theta$ that maximize the objective function $J(\theta)$. The function approximator estimate of the objective function then replaces the zeroth order parameter update from \eqref{equ:param_update_nocritic}, so that the actor parameter update becomes 
\begin{equation}
    \theta_{t+1 } = \theta_t + \eta \frac{F(\theta_t + \mu \mathbf{u}_t, \omega_t) - F(\theta_t - \mu \mathbf{u}_t, \omega_t)}{2\mu} \mathbf{u_t}, \label{eqn:zeroth_update_with_parameter}
\end{equation}
with $\eta > 0$ being the stepsize. The novel zeroth-order RL algorithm with function approximation is shown in Algorithm \ref{alg:ZOFA}. The sets $\mathcal{X}$ and $\mathcal{Y}$ are obtained by following the zeroth-order parameter sequence. That is, at each iteration of the algorithm, the parameters $\theta_t \pm \mu \mathbf{u}$ and the stochastic evaluation $\hat J (\theta_t \pm \mu \mathbf{u})$ are added to the sets $\mathcal{X},\mathcal{Y}$. These samples are \emph{imperishable}, as they can be continually sampled to update the parameter critic throughout the entire algorithm. Eventually, enough samples will be added so that with a sufficiently large batch size, \eqref{equ:non_convex} is obtainable and \eqref{equ:critic_bound} holds.

\begin{figure}[H]
	\centering	
	\includegraphics[scale=1]{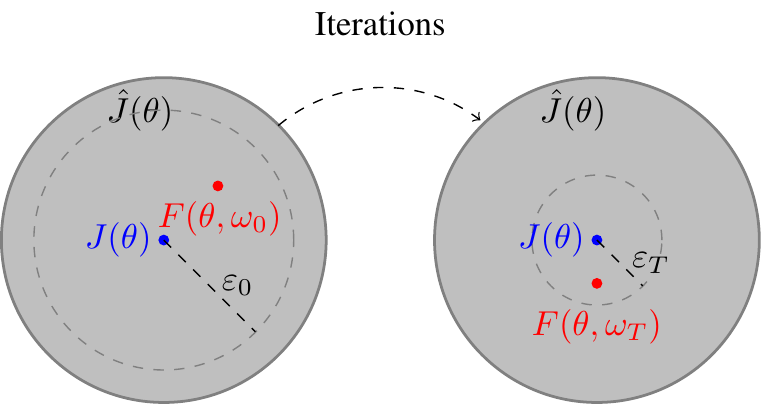}
	\caption{Illustration of the effect of the parameter critic in reducing the noise of the samples.}\label{fig:visualization}
\end{figure}

Figure \ref{fig:visualization} visually depicts the merit of using a global function approximator for zeroth-order parameter space optimization. Toward the beginning of the algorithm, there are not enough sample points in the sets $\mathcal{X}$ and $\mathcal{Y}$, which means that the error of the function approximator estimate may be large, as shown on the left. In contrast, as more samples are made available following the natural progression of the algorithm, the estimate becomes more accurate, as shown on the right. Our algorithm performs well once $\varepsilon$ is sufficiently smaller than the variance of the stochastic estimate.

In fact, the benefit of this method is twofold. Not only is the function estimate always improving, but also a record is maintained of the parameters visited in the past. This means that the algorithm may be reset to a globally optimal point in case a perturbation pushes the parameter outside a region of convergence. We now present our main theoretical result.

\begin{theorem} \label{thm:main}
Let Assumptions \ref{asm:bounded_reward} and \ref{asm:smoothness} be in effect.  Then, when the actor stepsize $\eta = T^{-1/2}$, it is true that the actor parameter sequence following Algorithm \ref{alg:ZOFA} satisfies
$$\frac{1}{T}\sum_{t = 1}^T \mathbb{E}\|\nabla J(\theta_t) \|^2 \leq \mathcal{O}\left( \frac{1}{T^{1/2}} + \frac{\bar{\varepsilon}_T}{\mu} + \mu + \frac{\bar{\varepsilon}_T^2}{T^{1/2}\mu^2}\right)$$
\end{theorem}
\begin{proof}
See the Appendix.
\end{proof}

Theorem \ref{thm:main} suggests that when $\mu$ is of order $\mathcal{O}(T^{-1/2})$ and $\bar \varepsilon_T $ is of order  $\mathcal{O}(T^{-3/2})$, a $\delta$ convergence to a stationary point is achievable in expectation with $\mathcal{O}(\delta^{-2})$ samples. While it is true that with a large enough network size, there exists a set of weights such that the function will be approximated with $\varepsilon$-accuracy, finding this set of parameters with noisy samples as in \eqref{equ:non_convex} may not be possible. Nevertheless, Theorem \ref{thm:main} provides an intuitive understanding on when a parameter critic will be able to solve the solution. The key comes from a closer inspection of the $\mu + \bar{\varepsilon}_T/\mu$ term. While in theory, we want $\mu$ to be as small as possible, values such as $\mu  = \{.1, .5, 1\}$ are known to not only converge but also lead the agent to meaningful solutions \cite{kumar2020zeroth}. This suggests that in practice, the function approximation average error only needs to be less than the zeroth-order perturbation $\mu$.

%%%\input{convergence_results}
%!TEX root = ACC_2020_parameter_critic.tex

\section{Numerical Results} \label{sec:Numericals}
In this section, we show empirical success on the use of a parameter critic. In particular, we first consider a trivial problem with varying sample noise to show where our method becomes useful. Second, we show that our method scales to higher dimensions by solving the canonical cartpole problem.  
\begin{figure*}[t]
	\centering
	\hspace{-1cm}
	\subfigure[$\sigma^2=0.01$.]{
		\centering
		\includegraphics{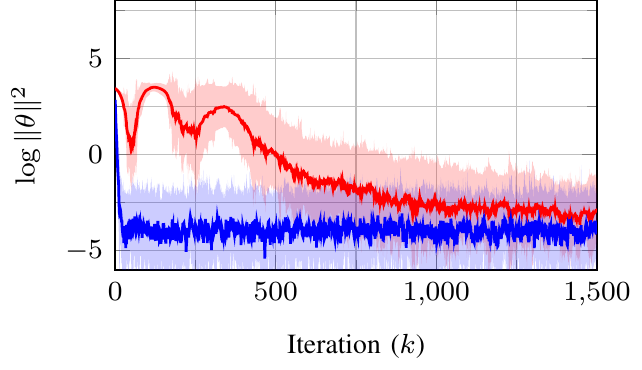}\label{fig:toy01}}
	\hspace{-0.85cm}
	\subfigure[$\sigma^2=0.1$.]{
		\centering
		\includegraphics{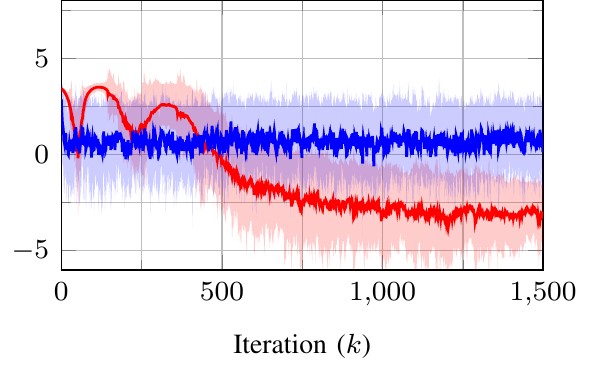} \label{fig:toy1} }
	\hspace{-1cm}	
	\subfigure[$\sigma^2=1$.]{
		\centering
		\includegraphics{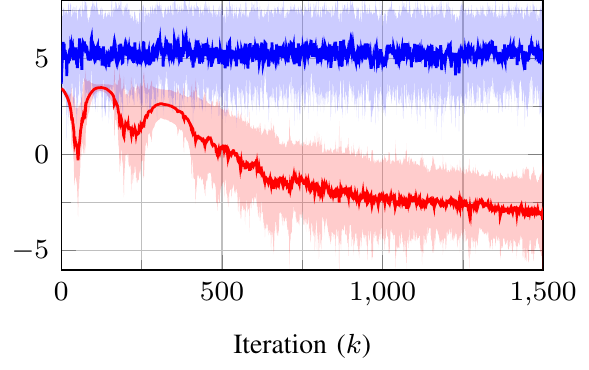} \label{fig:toy10} }
	\caption{Average of $\log \|\theta\|^2$ with confidence bounds over 50 trials with (red) and without (blue) the parameter critic. Learning rate, smoothing parameter, and parameter critic stepsize set to $\eta = 0.1$, $\mu = 0.005$, and $\alpha = 0.005$ respectively. Objective sample noise increases from left to right ($\sigma = \{.01, .1, 1\}$). }
	\label{fig:k_trend}
	\hspace{-1cm}
	\subfigure[Averaged value of the sampled objective function $\hat{J}(\theta)$.]{
		\centering
		\includegraphics{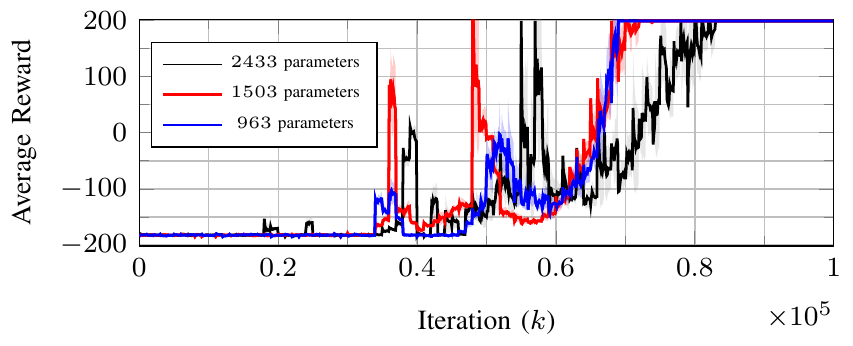}
		\label{fig:cartpole_actor}}
	\subfigure[MSE for the parameter critic for the case of the NN with $963$ parameters.]{
		\centering
		\includegraphics{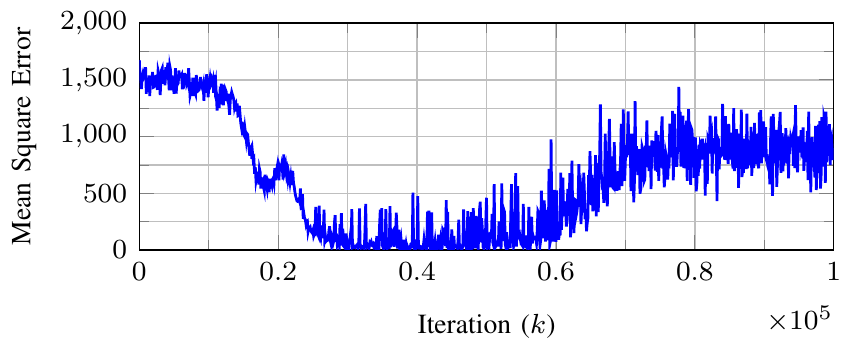}
		\label{fig:cartpole_critic}}
	\label{fig:cartpole}
	\caption{Cartpole simulations for three different actors parameterized by a two layer fully connected NN with $[64,32]$, $[54,22]$, and $[34,22]$ hidden units in the first and second layer respectively. Regarding the parameter critic, a three layer NN is used, with $256$, $128$, and $1$ neuron in the first, second and third layer respectively. }
	%\hspace{-1cm}
\end{figure*}

\subsection{When to use a Parameter Critic}
We first consider the trivial problem of minimizing the function 
\begin{equation*}
    J(\theta) = \|\theta\|^2,
\end{equation*}
where $\theta \in \mathbb{R}$. A parameter critic with a single hidden layer of $256$ neurons is learned on noisy samples drawn form a Gaussian distribution $\mathcal{N}(0, \sigma^2)$.  On each before each actor update step, a single sample is randomly drawn from the buffers $\mathcal{X}, \mathcal{Y}$ which is used to update the parameter critic. 

Figure \ref{fig:k_trend} compares the performance of Algorithm \ref{alg:ZOFA} with the finite difference parameter space update without a parameter critic \ref{equ:param_update_nocritic}. Shown in red, the update using the parameter critic convergence to $\theta$ is of the order of $\mathcal{O}(10^{-3})$ regardless of the sample noise. In contrast, the update without a parameter critic only outperforms Algorithm \ref{alg:ZOFA} in the case of figure \ref{fig:toy01} when the sample noise is small ($\sigma = 0.01$).  

We corroborate in practice the convergence of the parameter critic even in a noisy system. The convergence stems from the parameter critic's ability to use samples throughout the whole trajectory, gathering enough data to make the approximation error \eqref{equ:critic_bound} small enough thus reducing the variance of the zeroth order update \eqref{eqn:zeroth_update_with_parameter}. Using a parameter critic can be seen to be particularly valuable when the system is noisy, as only using samples of the system does not converge in practice. 

\subsection{Cartpole}
In this section, we benchmark our framework on a classical control problem, balancing a cartpole. This problem is featured in the OpenAI Gym toolkit \cite{1606.01540}. The state consists of the position and velocity of the cart and the angular position and velocity of the pole. The action space consists on the force applied to the cart. The initial state in randomly drawn, and a terminal episode consists of the angular position of the pole exceeding the limits of $[-12^{\circ},12^{\circ}]$ or the position of the cart $[-2.4,2.4]$. The rollouts consists of $200$ episodes, with a positive reward of $r=+1$ when the system is inside the bounds and $-1$ other wise. 

For both the actor and the parameter critic, two fully connected neural networks (NN) were used. In the case of the actor, the NN consists of two linear layers followed by a hyperbolic tangent. We benchmarked our framework for 3 different NN, with $[64,32]$, $[54,22]$, and $[34,22]$ hidden units in the first and second layers. In the case of the parameter critic, a three layer NN is used, with $256$, $128$, and $1$ neuron in the first, second and third layer respectively. In both cases, we are using ReLu as the non-linearity between layers. In the case of the actor, the parameters used for Algorithm \ref{alg:ZOFA} were actor learning rate $\eta=5^{-5}$, $\mu$-smoothing parameter is $\mu=0.5$ and  critic learning rate $\gamma=10^{-8}$. Due to memory limitations in practice, we resort to a usual practice of only storing the last $B=50$ rollout samples. Besides, in order to reduce the variance of the gradient of the critic $\hat{\nabla}_\omega \ell(\mathcal{Y},F(\mathcal{X},\omega)) $, we are using a mini-batch $b=20$ which is uniformly drawn from the buffer. 

%\begin{figure}[t]
%	\centering
%	\input{figures/Reward.tex}
%	\caption{Averaged value of the sampled objective function $\hat{J}(\theta)$ with respect to %iterations.}
%	\label{fig:Actor}
%\end{figure}

%\begin{figure}[t]
%	\centering
    % \includegraphics[scale=0.5]{figures/CriticMSE.pdf}        	
%    	\input{figures/Critic.tex}
%	\caption{The mean square error for the critic with respect to the iterations for the case of %the NN with $963$ parameters.}
%\label{fig:Critic}
%\end{figure}

% \begin{figure}[t]
% 	\centering
%     % \includegraphics[scale=0.5]{figures/CriticMSE.pdf}        	
% 	\caption{Toy example }
% \label{fig:toy_example}
% \end{figure}
% \begin{figure}[t]
% 	\centering
%     % \includegraphics[scale=0.5]{figures/CriticMSE.pdf}        	
%     	\input{figures/toy_examplep1.tex}
% 	\caption{Toy example 2}
% \label{fig:toy_examplep1}
% \end{figure}

% \begin{figure}[t]
% 	\centering
%     % \includegraphics[scale=0.5]{figures/CriticMSE.pdf}        	
%     	\input{figures/toy_examplep01.tex}
% 	\caption{Toy example 3}
% \label{fig:toy_examplep01}
% \end{figure}

In figure \ref{fig:cartpole_actor} we plot the performance of the agent, namely we plot $\hat{J}(\theta_k)$, with respect to the iteration for three NN of different sizes. Overall, all NN are capable of balancing the cartpole after $80000$ iterations. Furthermore, as expected, the smallest NN achieves the fastest convergence. This can be explained by the fact that a smaller NN for the actor, means a smaller subspace to be learned by the parameter critic.

In figure \ref{fig:cartpole_critic} we evaluate the mean square error between the parameter critic $F(\theta_t,\omega_t)$ and the noisy estimates of the rollouts $\hat{J}(\theta_k)$. In all, there are three distinctive phases in the learning trajectories. The first one takes place between the very beginning and 30 thousand iterations. Here the MSE of the parameter critic is high and the performance of the actor is poor, this suggests that the actor is not able to learn due to the error in the parameter critic. The second phase starts after the first one and ends at around 60 thousands iterations. In this phase the parameter critic overfits the noisy rollout samples $\hat{J}(\theta)$, thus the error is small, but the actor does not learn something meaningful as the zeroth order update is not precise. The last phase, starts at 60 thousand iterations and takes place once the parameter critic successfully learns the dependence between the objective function $J$ and the parameters $\theta$. The performance in the actor increases until it converges to $200$. Notice that the MSE for the parameter critic remains still, this is due to the noise in the samples, in statistics this error sometimes receives the name of irreducible error. 

%Two conclusions can be made, first the error made by the parameter critic does not converge to zero, which is related to the noise in the rollouts $\hat{J}$. Second, the positive slope in the actor curve is correlated with the one in the parameter critic. 
%\blue{This suggests that while the parameter critic is very small, the agent learns a good function approximation of a poor performing region of the parameter space. Moving from a poor performing region to a high performing region of the parameter space means seeing parameter samples that are new compared to what the parameter critic has seen before. As such, the error of the parameter critic increases. } %\red{Which means, that while the parameter critic has a very small error, that is not useful, as predicting the noisy estimates is not good for improving the actor policy. That is to say, once the critic is able to successfully predict the objective function $J(\theta)$, the error in the critic increases, and so does the performance of the actor. }
%!TEX root = ACC_2019.tex

\section{Future Directions and Conclusion} \label{sec:conclusion}

In this work, we introduced a variance reduction technique in which a function approximator is learned on the parameter space of an agent's policy. By interactions with the environment, samples of the performance of a policy are obtained and a parameter critic learns their dependence. We showed both theoretically and in practice how this approach is particularly useful when the system is noisy and rollouts are costly to obtain. Being able to evaluate the objective function directly has extensions beyond black-box reinforcement learning. In particular, constrained reinforcement learning is solvable through primal dual methods, though they require evaluations of the objective function in the dual update step \cite{paternain2019constrained}. These methods will also perform well in multi-agent reinforcement learning with Graph Neural networks, where the size of the network is smaller compared to the state space of the problem. Our method can also be extended to work in parallel with multiple agents communicating to obtain a better estimate of the parameter critic. We leave these avenues for future work.

\bibliographystyle{ieeetr}
\bibliography{bib-Nav_fun}

\begin{thebibliography}{10}

\bibitem{sutton2018reinforcement}
R.~S. Sutton and A.~G. Barto, {\em Reinforcement learning: An introduction}.
\newblock MIT press, 2018.

\bibitem{isele2018navigating}
D.~Isele, R.~Rahimi, A.~Cosgun, K.~Subramanian, and K.~Fujimura, ``Navigating
  occluded intersections with autonomous vehicles using deep reinforcement
  learning,'' in {\em 2018 IEEE International Conference on Robotics and
  Automation (ICRA)}, pp.~2034--2039, IEEE, 2018.

\bibitem{shalev2016safe}
S.~Shalev-Shwartz, S.~Shammah, and A.~Shashua, ``Safe, multi-agent,
  reinforcement learning for autonomous driving,'' {\em arXiv preprint
  arXiv:1610.03295}, 2016.

\bibitem{he2017deep}
Y.~He, Z.~Zhang, F.~R. Yu, N.~Zhao, H.~Yin, V.~C. Leung, and Y.~Zhang,
  ``Deep-reinforcement-learning-based optimization for cache-enabled
  opportunistic interference alignment wireless networks,'' {\em IEEE
  Transactions on Vehicular Technology}, vol.~66, no.~11, pp.~10433--10445,
  2017.

\bibitem{peters2003reinforcement}
J.~Peters, S.~Vijayakumar, and S.~Schaal, ``Reinforcement learning for humanoid
  robotics,'' in {\em Proceedings of the third IEEE-RAS international
  conference on humanoid robots}, pp.~1--20, 2003.

\bibitem{sutton2000policy}
R.~S. Sutton, D.~A. McAllester, S.~P. Singh, and Y.~Mansour, ``Policy gradient
  methods for reinforcement learning with function approximation,'' in {\em
  Advances in neural information processing systems}, pp.~1057--1063, 2000.

\bibitem{silver2014deterministic}
D.~Silver, G.~Lever, N.~Heess, T.~Degris, D.~Wierstra, and M.~Riedmiller,
  ``Deterministic policy gradient algorithms,'' in {\em Proceedings of the 31st
  International Conference on International Conference on Machine Learning -
  Volume 32}, ICML’14, p.~I–387–I–395, JMLR.org, 2014.

\bibitem{ghadimi2013stochastic}
S.~Ghadimi and G.~Lan, ``Stochastic first-and zeroth-order methods for
  nonconvex stochastic programming,'' {\em SIAM Journal on Optimization},
  vol.~23, no.~4, pp.~2341--2368, 2013.

\bibitem{vemula2019contrasting}
A.~Vemula, W.~Sun, and J.~Bagnell, ``Contrasting exploration in parameter and
  action space: A zeroth-order optimization perspective,'' in {\em The 22nd
  International Conference on Artificial Intelligence and Statistics},
  pp.~2926--2935, 2019.

\bibitem{paternain2018stochastic}
S.~Paternain, J.~A. Bazerque, A.~Small, and A.~Ribeiro, ``Stochastic policy
  gradient ascent in reproducing kernel hilbert spaces,'' {\em arXiv preprint
  arXiv:1807.11274}, 2018.

\bibitem{mania2018simple}
H.~Mania, A.~Guy, and B.~Recht, ``Simple random search provides a competitive
  approach to reinforcement learning,'' {\em arXiv preprint arXiv:1803.07055},
  2018.

\bibitem{zhang2019global}
K.~Zhang, A.~Koppel, H.~Zhu, and T.~Ba{\c{s}}ar, ``Global convergence of policy
  gradient methods to (almost) locally optimal policies,'' {\em arXiv preprint
  arXiv:1906.08383}, 2019.

\bibitem{ruckstiess2010exploring}
T.~R{\"u}ckstiess, F.~Sehnke, T.~Schaul, D.~Wierstra, Y.~Sun, and
  J.~Schmidhuber, ``Exploring parameter space in reinforcement learning,'' {\em
  Paladyn, Journal of Behavioral Robotics}, vol.~1, no.~1, pp.~14--24, 2010.

\bibitem{stulp2013robot}
F.~Stulp and O.~Sigaud, ``Robot skill learning: From reinforcement learning to
  evolution strategies,'' {\em Paladyn, Journal of Behavioral Robotics},
  vol.~4, no.~1, pp.~49--61, 2013.

\bibitem{bhatnagar2008incremental}
S.~Bhatnagar, M.~Ghavamzadeh, M.~Lee, and R.~S. Sutton, ``Incremental natural
  actor-critic algorithms,'' in {\em Advances in neural information processing
  systems}, pp.~105--112, 2008.

\bibitem{castro2010convergent}
D.~D. Castro and R.~Meir, ``A convergent online single time scale actor critic
  algorithm,'' {\em The Journal of Machine Learning Research}, vol.~11,
  pp.~367--410, 2010.

\bibitem{lillicrap2015continuous}
T.~P. Lillicrap, J.~J. Hunt, A.~Pritzel, N.~Heess, T.~Erez, Y.~Tassa,
  D.~Silver, and D.~Wierstra, ``Continuous control with deep reinforcement
  learning,'' {\em arXiv preprint arXiv:1509.02971}, 2015.

\bibitem{nesterov2017random}
Y.~Nesterov and V.~Spokoiny, ``Random gradient-free minimization of convex
  functions,'' {\em Foundations of Computational Mathematics}, vol.~17, no.~2,
  pp.~527--566, 2017.

\bibitem{kalogerias2019zeroth}
D.~S. Kalogerias and W.~B. Powell, ``Zeroth-order stochastic compositional
  algorithms for risk-aware learning,'' {\em arXiv preprint arXiv:1912.09484},
  2019.

\bibitem{kumar2020zeroth}
H.~Kumar, D.~S. Kalogerias, G.~J. Pappas, and A.~Ribeiro, ``Zeroth-order
  deterministic policy gradient,'' {\em arXiv preprint arXiv:2006.07314}, 2020.

\bibitem{kumar2019sample}
H.~Kumar, A.~Koppel, and A.~Ribeiro, ``On the sample complexity of actor-critic
  method for reinforcement learning with function approximation,'' {\em arXiv
  preprint arXiv:1910.08412}, 2019.

\bibitem{barron1993universal}
A.~R. Barron, ``Universal approximation bounds for superpositions of a
  sigmoidal function,'' {\em IEEE Transactions on Information theory}, vol.~39,
  no.~3, pp.~930--945, 1993.

\bibitem{cybenko1989approximation}
G.~Cybenko, ``Approximation by superpositions of a sigmoidal function,'' {\em
  Mathematics of control, signals and systems}, vol.~2, no.~4, pp.~303--314,
  1989.

\bibitem{eldan2016power}
R.~Eldan and O.~Shamir, ``The power of depth for feedforward neural networks,''
  in {\em Conference on learning theory}, pp.~907--940, 2016.

\bibitem{du2018gradient}
S.~Du, J.~Lee, Y.~Tian, A.~Singh, and B.~Poczos, ``Gradient descent learns
  one-hidden-layer cnn: Don’t be afraid of spurious local minima,'' in {\em
  International Conference on Machine Learning}, pp.~1339--1348, 2018.

\bibitem{ge2017learning}
R.~Ge, J.~D. Lee, and T.~Ma, ``Learning one-hidden-layer neural networks with
  landscape design,'' {\em arXiv preprint arXiv:1711.00501}, 2017.

\bibitem{sun2015nonconvex}
J.~Sun, Q.~Qu, and J.~Wright, ``When are nonconvex problems not scary?,'' {\em
  arXiv preprint arXiv:1510.06096}, 2015.

\bibitem{1606.01540}
G.~Brockman, V.~Cheung, L.~Pettersson, J.~Schneider, J.~Schulman, J.~Tang, and
  W.~Zaremba, ``Openai gym,'' 2016.

\bibitem{paternain2019constrained}
S.~Paternain, L.~Chamon, M.~Calvo-Fullana, and A.~Ribeiro, ``Constrained
  reinforcement learning has zero duality gap,'' in {\em Advances in Neural
  Information Processing Systems}, pp.~7555--7565, 2019.

\end{thebibliography}

\appendix
%!TEX root = ACC_2019.tex
\subsection{Auxiliary Lemmas and proofs}

% We first state and prove some important lemmas
% \begin{lemma} \label{lem:vemula} \cite[Lemma 4]{vemula2019contrasting} Let Assumptions \ref{asm:bounded_reward} and \ref{asm:smoothness} be in effect. Then it is true that $J_\mu(\theta)$ is also $G$-Lipshitz and $L$-smooth. Further, for all $\theta \in \Theta$, 
% $$\|\nabla_\theta J(\theta) - \nabla_\theta J_\mu (\theta) \| \leq L\mu.$$
% \end{lemma}

\begin{lemma} \label{lem:bound_bias} Let Assumptions \ref{asm:bounded_reward} and \ref{asm:smoothness} be in effect. Then, it is true that 
$$\|\nabla_\theta F_{\mu}(\theta_t,\omega_t) - \nabla J (\theta_t)\| \leq \varepsilon_t \mu^{-1}\sqrt{p}  + G\mu \sqrt{p}, $$ with $$\nabla_\theta F_{\mu}(\theta_t,\omega_t) := \mathbb{E}_\mathbf{u}\left[\frac{F(\theta_t + \mu \mathbf{u},\omega_t) - F (\theta_t - \mu \mathbf{u},\omega_t)}{2\mu}\mathbf{u}\right].$$
\end{lemma}

\begin{proof}
We can start by writing down a tautological equality, and then add and subtract $\nabla_{\theta}J_\mu(\theta_t)$ to the right hand side,
\begin{align*}
\nabla_\theta F_{\mu}(\theta_t,\omega_t) - \nabla J (\theta_t)=&\nabla_\theta F_{\mu}(\theta_t,\omega_t)-\nabla_{\theta}J_\mu(\theta_t)  \\ &+\nabla_{\theta}J_\mu(\theta_t) - \nabla J (\theta_t).
\end{align*}
    
% Let $\mathbb{E}_t$ denote the conditional expectation taken relative to history up to iteration $t-1$. Then 
% \begin{equation*}
% \begin{split}
% \mathbb{E}_t\left[g_t\right] &= \mathbb{E}_\mathbf{u}\left[\frac{J(\theta_t + \mu \mathbf{u}) - J (\theta_t - \mu \mathbf{u})}{2\mu}\mathbf{u}\right]\\
%     & +\mathbb{E}_\mathbf{u}\left[ \frac{J_{\omega_t}(\theta_t + \mu \mathbf{u}) - J(\theta_t + \mu \mathbf{u})}{2\mu}\mathbf{u}\right]\\
%     &+\mathbb{E}_\mathbf{u}\left[ \frac{ J(\theta_t - \mu \mathbf{u}) -J_{\omega_t}(\theta_t - \mu \mathbf{u}) }{2\mu}\mathbf{u}\right]
% \end{split}
% \end{equation*}
We can take norms, apply the triangle inequality, and by Lemma \ref{lem:zeroth_order_grad}, we obtain
\begin{equation*}
\begin{split}
&||\nabla_\theta F_{\mu}(\theta_t,\omega_t) - \nabla J (\theta_t)|| \\
&\leq||\nabla_\theta F_{\mu}(\theta_t,\omega_t)-\nabla_{\theta}J_\mu(\theta_t) || \\ & +||\nabla_{\theta}J_\mu(\theta_t) - \nabla J (\theta_t)||\\
&\leq\Big| \Big| \mathbb{E}_\mathbf{u}\left[\frac{F(\theta_t + \mu \mathbf{u},\omega_t) - F (\theta_t - \mu \mathbf{u},\omega_t)}{2\mu}\mathbf{u}\right] \\
&- \mathbb{E}_{\mathbf{u}}\left[ \frac{J(\theta + \mu \mathbf{u}) - J(\theta - \mu \mathbf{u})}{2\mu} \mathbf{u} \right] \Big| \Big| \\
&+|| \mathbb{E}_{\mathbf{u}}\left[ \nabla_{\theta}J(\theta_t+\mu \mathbf{u}) - \nabla J (\theta_t) \right]||.
\end{split}
\end{equation*}
Notice that this last line holds by the Monotone Convergence Theorem. We then apply the Jensen's inequality on the norm function, the $G$-smooth Assumption \ref{asm:smoothness}, function approximation bound $\varepsilon_t$ according to \eqref{equ:critic_bound}, and the fact that $\mathbb{E}_\textbf{u}[\|\textbf{u}\|] \leq \sqrt{p}$ as $\textbf{u}\sim \mathcal{N}(\mathbf{0},I_p)$, to obtain, 
\begin{align*}
||\nabla_\theta F_{\mu}(\theta_t,\omega_t) - \nabla J (\theta_t)|| \leq&\Big ( \frac{\varepsilon_t} {\mu} +G \mu \Big ) \mathbb{E}_\mathbf{u}\left[\Big| \Big| \mathbf{u} \Big| \Big|\right] \\
\leq& \varepsilon_t \mu^{-1}\sqrt{p}  + G\mu \sqrt{p}. 
\end{align*}
\end{proof}

\begin{lemma} \label{lem:bound_variance} Let Assumptions \ref{asm:bounded_reward} and \ref{asm:smoothness} be in effect. Then it is true that 
\begin{equation} \label{equ:lem3}
    \begin{split}
        \mathbb{E}\Bigg\|\frac{F(\theta_t + \mu \mathbf{u}_t, \omega_t) - F(\theta_t - \mu \mathbf{u}_t,\omega_t)}{2\mu}&\mathbf{u_t}\Bigg\|^2 \\\leq \frac{10}{3}\left(Lp^2+\frac{\varepsilon_t^2 p}{\mu^2}\right).\\
    \end{split}
\end{equation}
%$$\mathbb{E}\left\|\frac{F(\theta_t + \mu \mathbf{u}_t, \omega_t) - F(\theta_t - \mu \mathbf{u}_t,\omega_t)}{2\mu}\mathbf{u_t}\right\|^2 \leq \frac{10}{3}\left(Lp^2+\frac{\varepsilon^2 p}{\mu^2}\right).$$
\end{lemma}

\begin{proof}
%By definition, we write
%\begin{equation*}
 %   \nabla_\theta F_{\mu}(\theta_t,\omega_t) = \mathbb{E}_\mathbf{u}\left[\frac{F(\theta_t + \mu \mathbf{u},\omega_t) - F (\theta_t - \mu \mathbf{u},\omega_t)}{2\mu}\mathbf{u}\right].
%\end{equation*}
For brevity in notation, let $\hat \nabla_\theta F_\mu(\theta_t, \omega_t)$ stand for the fradient estimate at time $t$ given on left hand side of \eqref{equ:lem3} inside the norm squared. Add and subtract $J(\theta_t + \mu \mathbf{u}) + J(\theta_t - \mu \mathbf{u})$, and rearrange to obtain
\begin{equation*}
\begin{split}
    \hat \nabla_\theta F_{\mu}(\theta_t,\omega_t) &= \frac{J(\theta_t + \mu \mathbf{u}_t) - J (\theta_t - \mu \mathbf{u}_t)}{2\mu}\mathbf{u}_t\\
    & + \frac{F(\theta_t + \mu \mathbf{u}_t,\omega_t) - J(\theta_t + \mu \mathbf{u}_t)}{2\mu}\mathbf{u}_t\\
    &+ \frac{ J(\theta_t - \mu \mathbf{u}_t) -F(\theta_t - \mu \mathbf{u}_t,\omega_t) }{2\mu}\mathbf{u}_t.
\end{split}
\end{equation*}
Take the norm squared of both sides and apply $\|a + b + c\|^2 \leq 10\left( \| a\|^2 + \|b\|^2 + \|c\|^2 \right)/3$ to obtain
\begin{equation*}
\begin{split}
    \|\hat \nabla_\theta F_{\mu}(\theta_t,&\omega_t)\|^2  \\
    &\leq \frac{10}{3}\left\|\frac{J(\theta_t + \mu \mathbf{u}_t) - J (\theta_t - \mu \mathbf{u}_t)}{2\mu}\mathbf{u}_t\right\|^2 \\
    & + \frac{10}{3}\left\|\frac{F(\theta_t + \mu \mathbf{u}_t, \omega_t) - J(\theta_t + \mu \mathbf{u}_t)}{2\mu}\mathbf{u}_t\right\|^2\\
    &+\frac{10}{3} \left\| \frac{ J(\theta_t - \mu \mathbf{u}_t) -F(\theta_t - \mu \mathbf{u}_t,\omega_t) }{2\mu}\mathbf{u}_t\right\|^2.
    \end{split}
\end{equation*}
Take the expectation of both sides. Finally, we bound the expected value of the square of the norm of the normal vector $\mathbb{E}[\|\mathbf{u}\|^2]=p^2$, and we bound the first term by Lipschitz Assumption \ref{asm:smoothness}, and the second and third terms by function appproximation error $\varepsilon_t$ according to \ref{equ:critic_bound} to obtain the final result,
\begin{equation*}
    \mathbb{E}\| \hat \nabla_\theta F_{\mu}(\theta_t,\omega_t)\|^2  \leq \frac{10}{3}\left(L^2p^2+\frac{\varepsilon_t^2 p}{2\mu^2}\right).
\end{equation*}
\end{proof}

Now we are ready to prove the main theorem.

\subsection{Proof of Theorem \ref{thm:main}} 

Using the assumption that $J(\theta)$ is $L$-smooth, we have 
\begin{equation*}
 J(\theta_{t+1}) \geq J(\theta_t) + (\theta_{t+1} - \theta_t)^\top \nabla J(\theta_t) - L\|\theta_{t+1}- \theta_t\|^2.
\end{equation*}
Let $\mathbb{E}_t$ denote the conditional expectation taken relative to history up to iteration $t-1$. Further, let $\eta \nabla_\theta F_{\mu}(\theta_t,\omega_t)$ denote the expected update $\theta_{t-1} - \theta_t$ as shown in algorithm \ref{alg:ZOFA}. Then, we obtain
\begin{align*}
    \mathbb{E}_t &\left[J(\theta_{t+1})\right] \geq  J(\theta_t) +\eta (\nabla_\theta F_{\mu}(\theta_t,\omega_t))^\top \nabla J(\theta_t) \\
    &- L \eta^2 \mathbb{E}_t\left\|\frac{F(\theta_t + \mu \mathbf{u}_t, \omega_t) - F(\theta_t - \mu \mathbf{u}_t,\omega_t)}{2\mu}\mathbf{u_t}\right\|^2.
\end{align*}
Add and subtract $\eta \nabla J(\theta_t)^\top \nabla J(\theta)$ and apply Lemma \ref{lem:bound_variance} to obtain
\begin{align*}
    \mathbb{E}_t &\left[J(\theta_{t+1})\right] \geq\\
    &J(\theta_t)+ \eta \|\nabla J(\theta_t)\|^2 - L \eta^2 \frac{10}{3}\left(L^2 p^2 + \frac{\varepsilon_t^2p}{2\mu^2}\right)  \\
    &+\eta (\nabla_\theta F_{\mu}(\theta_t,\omega_t) - \nabla J (\theta_t))^\top \nabla J(\theta_t).
\end{align*}
%
%\begin{equation*}
%\begin{split}
%    \mathbb{E}_t \left[J(\theta_{t+1})\right] &\geq J(\theta_t) + \eta \left(\nabla_\theta F_\mu(\theta_t,\omega_t) - \nabla J(\theta_t)\right)^\top \nabla J(\theta_t)\\
 %   & + \eta \|\nabla J(\theta_t)\|^2- L \eta^2 \mathbb{E}_t\left[\|g_t\|^2\right].
%\end{split}
%\end{equation*}
%
Using Cauchy-Schwartz and the fact that $J(\theta)$ is $G$-Lipshitz, and Assumption \ref{asm:bounded_reward}, to bound  $\left(\nabla F_\mu(\theta_t, \omega_t) - \nabla J(\theta_t)\right)^\top \nabla J(\theta_t)$ by $-\beta\|\nabla F(\theta_t, \omega_t) - \nabla J(\theta_t)\|$ to obtain
\begin{align*}
    \mathbb{E}_t &\left[J(\theta_{t+1})\right] \geq\\
    &J(\theta_t)+ \eta \|\nabla J(\theta_t)\|^2 - L \eta^2 \frac{10}{3}\left(L^2 p^2 + \frac{\varepsilon_t^2p}{2\mu^2}\right)  \\
    &-\eta \beta\|\nabla_\theta F_{\mu}(\theta_t,\omega_t) - \nabla J (\theta_t)\|.
\end{align*}
%
%\begin{equation*}
%   \begin{split}
%    \mathbb{E}_t \left[J(\theta_{t+1})\right] &\geq J(\theta_t) - \eta G\|g_t - \nabla J(\theta_t)\|\\
%    & + \eta \|\nabla J(\theta_t)\|^2- L \eta^2 \mathbb{E}_t\left[\|g_t\|^2\right].
%\end{split}
%\end{equation*}
%
Take the total expectation and invoke Lemma \ref{lem:bound_bias}  to obtain 
\begin{equation*}
    \begin{split}
    \mathbb{E} \left[J(\theta_{t+1})\right] &\geq J(\theta_t) - \eta \beta \sqrt{p}\left(\varepsilon_t\mu^{-1} + G\mu \right)\\
    & + \eta \mathbb{E}\|\nabla J(\theta_t)\|^2 - L \eta^2 \frac{10}{3}\left(L^2 p^2 + \frac{\varepsilon_t^2p}{2\mu^2}\right) .
    \end{split}
\end{equation*}
Reorder terms and sum until $T$ and divide by $1/(\eta T)$
\begin{equation*}
    \begin{split}
        \frac{1}{T} \sum_{t = 1}^{T} \mathbb{E} \|\nabla J (\theta_t)\|^2 & \leq \frac{J(\theta_0) - \mathbb{E}\left[J(\theta_T)\right]}{\eta T}  + 2\beta\bar{\varepsilon}_T\mu^{-1}\sqrt{p} \\
         & +G \beta\mu + \frac{10}{3}L^3 \eta p^2 + \frac{10}{3}L\eta \bar{\varepsilon}^2_T p \mu^{-2}.
    \end{split}
\end{equation*}
Bound $J(\theta_0) - \mathbb{E}\left[J(\theta_T)\right]$ by $2\mathcal{J}$ by Assumption \ref{asm:bounded_reward}, define $\bar{\varepsilon}_T=\sum_{t=1}^T \varepsilon_t/T$. Set $\eta = T^{-1/2}$ to conclude the proof.
\end{document}